\documentclass[10pt,twocolumn,aps,prl,superscriptaddress,floatfix]{revtex4}

\usepackage{braket}
\usepackage{epsfig}
\usepackage{mathtools}
\usepackage{amsmath, amssymb, amsfonts, mathrsfs}
\usepackage{amsthm}
\usepackage{dsfont} 
\usepackage{graphicx}
\usepackage{color}
\usepackage{xcolor}
\usepackage{tcolorbox}
\usepackage{hyperref}
\usepackage{changes}
\hypersetup{
    colorlinks,
    linkcolor={rgb:red,2; blue,1},
    citecolor={blue!60!black},
    urlcolor={blue!60!black}
}
\usepackage{comment}
\usepackage{bbm}
\usepackage{soul}
\usepackage{bm}

\tcbset{colback=orange!5,colframe=orange!50!blue,
fonttitle=\bfseries, float=htb}
\usepackage[paperwidth=210mm,paperheight=297mm,centering,hmargin=2.4cm,vmargin=2cm]{geometry}

\DeclareMathOperator{\Tr}{Tr}
\DeclareMathOperator*{\argmax}{arg\,max} 

\newtheorem{theorem}{Theorem}
\newtheorem{corollary}{Corollary}
\newtheorem{lemma}{Lemma}

\newtheorem{claim}{Claim}
\theoremstyle{definition}
\newtheorem{definition}{Definition}
\usepackage{slashbox} 
\usepackage[framemethod=tikz]{mdframed}
\definecolor{mycolor}{rgb}{0.122, 0.435, 0.698}
\newmdenv[innerlinewidth=0.5pt, roundcorner=4pt,linecolor=mycolor,innerleftmargin=6pt,
innerrightmargin=6pt,innertopmargin=6pt,innerbottommargin=6pt]{mybox}
\usepackage{paracol}
\usepackage{tcolorbox}
\tcbuselibrary{theorems}

\newtcbtheorem{Theorem}{Theorem}%
{colback=mycolor!5,colframe=mycolor,fonttitle=\bfseries,
left=4pt,
right=4pt,
top=4pt,
bottom= 4pt
}{thm}

\begin{document}

\title{Unifying paradigms of quantum refrigeration:\\A universal and attainable bound on cooling}
\author{Fabien Clivaz}
\affiliation{Department of Applied Physics, University of Geneva, 1211 Geneva 4, Switzerland}
\affiliation{Institute for Quantum Optics and Quantum Information (IQOQI), Austrian Academy of Sciences, Boltzmanngasse 3, A-1090 Vienna, Austria}
\author{Ralph Silva}
\affiliation{Department of Applied Physics, University of Geneva, 1211 Geneva 4, Switzerland}
\affiliation{Institute for Theoretical Physics, ETH Z\"urich, 8093 Z\"urich, Switzerland}
\author{G\'eraldine Haack}
\affiliation{Department of Applied Physics, University of Geneva, 1211 Geneva 4, Switzerland}
\author{Jonatan Bohr Brask}
\affiliation{Department of Applied Physics, University of Geneva, 1211 Geneva 4, Switzerland}
\affiliation{Technical University of Denmark, 2800 Kgs. Lyngby, Denmark}
\author{Nicolas Brunner}
\affiliation{Department of Applied Physics, University of Geneva, 1211 Geneva 4, Switzerland}
\author{Marcus Huber}
\affiliation{Institute for Quantum Optics and Quantum Information (IQOQI), Austrian Academy of Sciences, Boltzmanngasse 3, A-1090 Vienna, Austria}

\begin{abstract}
Cooling quantum systems is arguably one of the most important thermodynamic tasks connected to modern quantum technologies and an interesting question from a foundational perspective. It is thus of no surprise that many different theoretical cooling schemes have been proposed, differing in the assumed control paradigm and complexity, and operating either in a single cycle or in steady state limits. Working out bounds on quantum cooling has since been a highly context dependent task with multiple answers, with no general result that holds independent of assumptions. In this letter we derive a universal bound for cooling quantum systems in the limit of infinite cycles (or steady state regimes) that is valid for any control paradigm and machine size. The bound only depends on a single parameter of the refrigerator and is theoretically attainable in all control paradigms. For qubit targets we prove that this bound is achievable in a single cycle and by autonomous machines.
\end{abstract}

\maketitle 

Characterising the ultimate performance limits of quantum thermal machines is directly related to the understanding of energy exchanges at the quantum scale, and hence to the formulation of thermodynamic laws valid in the quantum regime \cite{review}. A central challenge towards this goal is to identify the relevance of different levels of control that can be achieved over quantum mechanical machines and how this determines the thermodynamic limits in the quantum regime. To analytically derive such fundamental limits, it is instructive to observe that the time evolution of closed quantum systems is unitary in quantum mechanics. Discussing non-trivial changes of energy and entropy on a target system thus requires a conceptual separation of a thermodynamic process into a target system, upon which a desired task is performed, and a `machine', comprised of several other quantum systems, which is used to perform the task upon the target. The total evolution will be a global unitary operation and feature a limited depth/complexity (i.e. machines and the baths constituents they couple to have a finite size) and can either be assumed to be energy changing, if one allows for a \emph{coherent} control of the process, or alternatively energy preserving in an \emph{incoherent} control paradigm. Furthermore, machines can operate in a cyclic manner, such that each of these unitary operations can be repeated an arbitrary number of times. To derive self-contained and non-trivial bounds on thermodynamic performance, the state of the machine for each of the repetitions should be thermal with respect to its Hamiltonian. Additionally, to include heat engines, part of the machine could be thermalised with respect to a higher temperature. All these different paradigms are illustrated in Figure \ref{fig:model}, and include paradigmatic scenarios, such as autonomous quantum thermal machines \cite{virtual, auto1,auto2}, heat bath algorithmic cooling \cite{algo1,algo2,algo3,algo4,algo5,algo6}, and quantum otto engines \cite{otto1,otto2} for finite size processes. Our framework is a special case of the generalised framework exposed in \cite{alvaro} but differs from any of the case studies made there. In the limit of infinite size, the limiting cases are the resource theory of thermodynamics \cite{resource,TO,work,work2} for incoherent operations and for coherent control general completely positive and trace preserving (CPTP) maps, that make cooling to the ground state trivially possible, emphasizing the role of complexity in the third law of thermodynamics \cite{thirdlaw,thirdlaw2}.
	\begin{figure}[t]
		\centering
		\includegraphics[width=7.5cm]{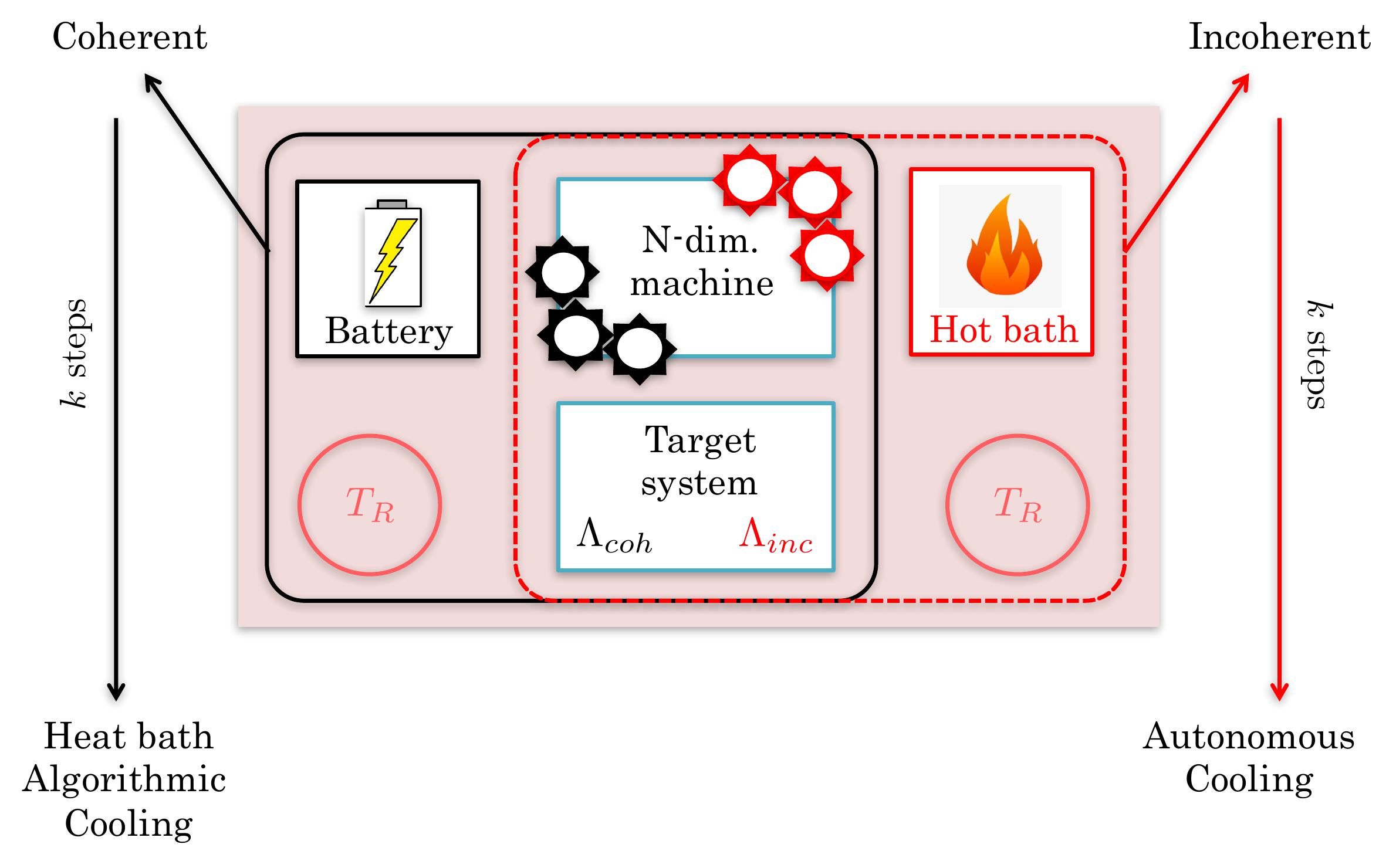}
		\caption{Schematic representation of the paradigms of coherent and incoherent control for quantum refrigeration. For a target qubit, the coldest state of the infinite cycle regime of the incoherent paradigm is the steady state achieved by autonomous cooling. For a target qubit and a product qubit machine, the coherent paradigm is a special instance of heat bath algorithmic cooling with no compression qubit.}\label{fig:model}
	\end{figure}

A task of paramount importance for quantum technologies is refrigeration and, depending on the paradigm, different limitations have been derived for specific machine designs. Indeed, it seems a daunting task to derive bounds beyond specific, low complexity machines, as the potential Hamiltonians, dimensions and number of repetitions present an overwhelming amount of parameters to optimise over. Thermodynamics as a theory, however, has been astoundingly successful in deriving general statements that are independent of the microscopic complexity, by identifying a small number of relevant parameters that ultimately determine the limits of processing energy.

We should first clarify that cooling a quantum system can have several meanings. For a system initially in a thermal state, one can drive it to a thermal state of lower temperature. Alternatively, since for some paradigms fixing the entire spectrum is actually too strong a condition, one could just consider increasing its ground-state population or its purity, or decreasing its entropy or its enerzgy (see e.g. discussions on passivity \cite{passive1,passive2,passive3,passive4}). These notions are in general nonequivalent for target systems of arbitrary dimension and determining the fundamental limits to cooling is therefore dependent on the choice of target function. {However, by using majorisation theory, we are able to derive results that hold for all of the above mentioned notions of cooling.}

\emph{Setting.} We consider a target system of dimension $d_S$ with Hamiltonian $H_S = \sum_{i=0}^{d_S-1} E_i \ket{E_i}_S \bra{E_i}$, where $E_i \leq E_j$ for $i <j$ and a thermodynamic machine of finite size $d_M$ with Hamiltonian $H_M= \sum_{i=0}^{d_M-1} \mathcal{E}_i \ket{\mathcal{E}_i}_M \bra{\mathcal{E}_i}$, where $\mathcal{E}_i \leq \mathcal{E}_j$ for $i<j$, both surrounded by a thermal bath at temperature $T_R$ (or inverse temperature $\beta_R$).  Without loss of generality we assume $E_0=\mathcal{E}_0=0$. We also write $\mathcal{E}_{\text{max}}=\mathcal{E}_{d_M-1}$ and for a qubit target, $d_S=2$, $E_S=E_1$. The joint initial state of target and machine is given by $\rho_{SM}=\tau_S(\beta_R)\otimes\tau_M(\beta_R)$, where $\tau(\beta_R)$ denotes a thermal state at inverse temperature $\beta_R$.
Our goal is to cool the target system. After the application of a single unitary the target state is changed as $\text{Tr}_M[U\rho_{SM}U^\dagger]=:\Lambda(\tau_S(\beta_R))=:\rho_S'$. To induce a non-trivial change of the target system,  out-of-equilibrium resources have to be used, which we conceptually divide into two extremal cases: Either we use external coherent control, to induce any energy changing unitary of target and machine (see e.g. \cite{coh1}), inducing the map $\Lambda_{coh}$ on the system. Or we consider an energetically closed system, i.e. consider unitaries that commute with the total system-machine Hamiltonian $H=H_S+H_M$, inducing $\Lambda_{inc}$ on the target.
This application of a unitary represents a single cycle of the machine. Between cycles the machine is rethermalised, where individual machine components can thermalise either with the environment at temperature $T_R$ or draw further resources from a hot bath at temperature $T_H$. For many cycles we denote the number of cycles $k$ by a superscript on the map $\Lambda^k_{coh/inc}$. Of particular interest to us is the unbounded cycle regime, i.e. as $k \rightarrow \infty$, which we denote by $p_0^*$. The setting is visually represented in Fig. \ref{fig:model} and the asymptotic relations of the induced maps is summarised in Table \ref{table} in the Appendix.\\

\emph{Results.} In the following we consider an arbitrary sequence of operations and prove a bound that holds for any control paradigm and any machine size in the limit of infinite cycles. This bound only depends on the maximal energy gap of the machine, irrespective of the structure of $H_M$. The bound can be attained for both coherent and incoherent machines. Finally, for qubit targets, the bound can even be obtained in a single cycle by the smallest possible coherent machine (i.e. a single qubit machine), as well as via the smallest autonomous refrigerator (i.e. a two-qubit machine).

The energetics of quantum systems only depend on the diagonal elements in the energy eigenbasis, which through unitary evolution change uni-stochastically. This already gives a recipe for single cycle cooling that is analytically optimal in terms of majorization theory: The global diagonal of the joint target and machine density matrix of the final state is always majorized by the initial diagonal (Shur-Horn Lemma). Thus, the optimal unitary is one that puts all the largest eigenvalues in the diagonal entries that contribute to the ground state population of the target, {{red} the next largest in those entries contributing to the first excited state of the target, and so on}. Which unitary that is, depends on all energy gaps and respective eigenvalues for each cycle.

Descriptions of thermodynamic machines, however, should go beyond such optimal single cycles and also allow for a repetition of cycles to determine its ultimate limits. Ideally, each cycle can be perfectly separated and the machine rethermalized in between cycles. In both coherent and incoherent paradigms we thus also consider the limit of infinite repetitions of optimal unitary operations, where the machine can be perfectly rethermalized inbetween. While optimal unitary operations can at least in principle be described for every cycle of arbitrary machine and target systems, they will in general depend on the entire spectrum and does not identify relevant machine parameters for determining performance. Indeed, determining simple and universal bounds on reachable target temperatures and revealing the relevant parameters of machines in all paradigms, even after infinite cycle repetitions, is the main result of this letter. We present four theorems gauging the cooling performances for the different types of machines, prove universal upper bounds on cooling and demonstrate their respective attainability. In the accompanying article \cite{ourPRE}, we furthermore investigate the work cost of achieving these bounds. 

\textit{Universal bound on cooling.}  Unless stated otherwise, we consider an arbitrary machine, i.e. any spectrum and size. It turns out however that no matter the microscopic physics of the machine, only a single machine parameter determines the ultimate cooling bound, namely the largest energy gap $\mathcal{E}_{\text{max}}$.

\begin{theorem}{(Universal bound on cooling)}\label{thm:unibound}
For any machine and control paradigm, in the limit of infinite cycles,
\begin{itemize}
    \item for a qubit target, the ground state population is upper bounded by 
    \begin{align}\label{math:bestgroundpop}
        p_0^* &= \frac{1}{1 + e^{-\beta_R \mathcal{E}_\text{max}}},
    \end{align}
    \item for an arbitrary target of dimension $d_S$, the vector of eigenvalues of the final state is majorized by that of the following state,
    \begin{align}\label{math:beststate}
        \rho^*_S = \frac{1}{\mathcal{Z}} \sum_{n=0}^{d_S-1} \left(e^{- \beta_R \mathcal{E}_{max}}\right)^n \ket{E_n}_S\!\bra{E_n},
    \end{align}
    as long as the initial state is majorized by $\rho^*_S$. In particular, the ground state population, purity, entropy and average energy of the final state are bounded by those of $\rho^*_S$.
\end{itemize}
\end{theorem}

Before we discuss the proof, we would like to make a few remarks. The qubit bound, first derived in \cite{armen}, and also appearing in \cite{NJP2014} corresponds to an inverse temperature of
\begin{equation}\label{math:besttempqubit}
    \beta^* = \frac{ \beta_R \mathcal{E}_{max}}{E_S}\,.
\end{equation}
The crucial parameter is the ratio between the populations of the excited and ground state, $g = e^{-\beta_R \mathcal{E}_\text{max}}$. The bound \eqref{math:beststate} for higher dimensional systems has the same property, that the ratio of populations for every pair of successive energy levels is given by $g$.

We emphasize the advantage of a bound based on majorisation. Since the state \(\rho_S^*\) is the unique \emph{passive} state \cite{footnote1} that majorizes all others also attainable by coherent operations, it upper (lower) bounds every Schur convex (concave) function of the eigenvalues, which includes the various notions of cooling listed in the theorem.

 \begin{proof}[Proof of Theorem \ref{thm:unibound}] Firstly, for machines using coherent operations, using the Schur-Horn Lemma for the joint state of target and machine $\rho_{SM}$, one can verify (see the appendix) that the system after a single cycle operation, \(\rho_S'\), satisfies \(\rho_S' \prec \rho_S^*\) as long as \(\rho_S \prec \rho_S^*\). By induction, one arrives at \(\rho_S^{(n)} \prec \rho_S^*\) for all $n$.
 
 Furthermore, note that a hotter thermal state is always majorized by a colder thermal state and majorization is stable under tensor products, see Corollary 1.2. of Ref~\cite{bondar}. Thus having access to a hot thermal bath to rethermalize our machine before a cycle, or part of it, if the machine has a tensor product structure, will only enable us to reach states that are majorized by the initial state. Thus a hot thermal bath does not allow us to reach a state that cannot already be reached by a coherent resource, completing the proof.
\end{proof}

\textit{Attaining the cooling bound--} The bound defined in Eq.\eqref{math:beststate} only depends on the subspace of the machine with the ``coldest" ratio of populations. Indeed, one can construct an explicit cooling protocol as follows: consider the following simple joint unitary operation between target and machine,
\begin{equation}
\begin{aligned}
    U_i = \mathds{1} &- \ket{E_{i-1} \mathcal{E}_{\text{max}}}\! \bra{E_{i-1} \mathcal{E}_{\text{max}}} - \ket{E_{i} \mathcal{E}_{0}} \! \bra{E_{i} \mathcal{E}_{0}}\\
    &  + \ket{E_{i-1} \mathcal{E}_{\text{max}}} \! \bra{E_{i} \mathcal{E}_{0}} + \ket{E_{i} \mathcal{E}_{0}} \! \bra{E_{i-1} \mathcal{E}_{\text{max}}}.
    \end{aligned}
\end{equation}

Each $U_i$ is a simple qubit swap, between the $i^{th}$ pair of successive energy eigenstates of the target, and the $\mathcal{E}_\text{max}$ subspace of the machine, and will lead to the population of the lower energy state $p_{i-1}$ increasing by some $\Delta_i$ (this may be negative). We define the \textit{``coherent max-swap" operation} as the one that performs the $U_i$ corresponding to the greatest positive value of $\Delta_i$ (if none exists, then no unitary is performed). The choice and implementation of $U_i$ is preceded and followed by a unitary on the target that leaves the target in a passive state.

\begin{theorem}{(Coherent attainability)}\label{thCoattain}
The state of the target under the repeated application of either one of the coherent max-swap operation or the optimal coherent operation converges to $\rho^*_S$.
\end{theorem}

\begin{proof}[Proof.]

Since the state of the target is passive before a particular $U_i$ is performed, and $\Delta_i>0$ in this case, the population is always moved from a smaller eigenvalue to a larger one. Thus the final state always majorizes the initial one. This is also true for the optimal coherent operation, since by construction it leaves the target in a state that majorizes all others. Thus the ordered partial sums of eigenvalues of the target under repeated application of either operation form monotonically increasing sequences. As these sums are bounded by one, both protocols must converge.

In the appendix, we show that the state converged to is a fixed point of the coherent max-swap operation, and that all fixed points of this protocol have the property of majorizing $\rho^*_S$. By Theorem \ref{thm:unibound}, the final state under any protocol must be majorized by $\rho^*_S$, proving that the only possible convergent point of the max-swap is $\rho^*_S$ itself.

The proof for the case of optimal coherent operations follows from the fact that the state under many cycles majorizes the state under an equal number of cycles of any other coherent operations (see appendix). Thus the state converged to must majorize that of the max-swap, but also still be majorized by $\rho^*_S$ (Theorem \ref{thm:unibound}). Therefore it also converges to $\rho^*_S$.

\end{proof}

This shows that the universal bound is tight in the coherent case. Since the operations involved are non-energy conserving, one cannot conclude the same for the incoherent paradigm. Nonetheless, a minor addition to the machine enables a similar statement.

In the coherent case we only required to swap qubit subspaces in the target with the maximum energy gap of the machine. Each of these swaps can be made energy-preserving by adding one more qubit to the machine to bridge the energy difference, and thermalising this qubit to a hot temperature $T_H > T_R$ in between cycles. This defines an \textit{incoherent} version of the max-swap protocol.

\begin{theorem}{(Incoherent attainability)}\label{thInattain}
In the limit $T_H \rightarrow \infty$, one can incoherently cool the target to at least the coherent cooling bound of $\rho^*_S$ in the infinite cycle limit, if one extends the machine by (at most) $d_S-1$ qubits, each with energy gap $\mathcal{E}_\text{max} - (E_i - E_{i-1})$, $i \in \{1,2,...,d_S-1\}$, or more generally, any extension that has qubit subspaces with these energy gaps.
\end{theorem}

\begin{proof}

The proof is analogous to that of Theorem 2. Here too, the (incoherent) max-swap operation has the property of the final state always majorising the initial one. It follows that the repeated application converges. One can prove (see appendix) that the convergent point is a fixed point of the operation, and that in the limit $T_H \rightarrow \infty$, the fixed points all majorise $\rho^*_S$, proving attainability as desired.

\end{proof}
Theorem \ref{thInattain} shows that the two paradigms are in fact closely related to one another in terms of cooling performance. In the limit of \(T_H \rightarrow \infty\), a finite resource  of the same dimension of the target suffices to bridge the gap between them.

\textit{Bridging extremal cooling paradigms.}  So far, we dealt exclusively with stroke type machines, in the sense that we allowed for rethermalization or unitary operations in well separated discrete time steps. One may, however, wonder if the bounds on cooling set by Theorems \ref{thm:unibound}-\ref{thInattain} are also valid for autonomous thermal machines, where thermalization and machine cycles happen simultaneously and continuously.
Interestingly, a link exists between those machines and our incoherent paradigm, see also \cite{Raam}. In essence, the energy-preserving unitaries of the incoherent paradigm are replaced by a time-independent energy-preserving interaction Hamiltonian between the target system and the machine. It is then straightforward to prove the following statement.

\begin{theorem}{(Autonomous machines correspondence)}\label{th4} Consider an arbitrary machine and a qubit target system. When the target coupling to the thermal bath is zero, one can cool the target to \(\beta^*\) with an autonomous machine by extending the machine as in Theorem \ref{thInattain}.
\end{theorem} 

\begin{proof}
This follows directly from known results for small autonomous machines. Following the derivation in \cite{auto0}, which is based on a linear master equation, one can replace the second qubit by the subnormalised qubit \(\rho_{0,\text{max}}=\frac{1}{\mathcal{Z}}(\ket{0}\bra{0}+ e^{-\beta_R\mathcal{E}_{\text{max}}} \ket{\mathcal{E}_{\text{max}}} \bra{\mathcal{E}_{\text{max}}})\). Furthermore, adding energy levels to the machine between \(\ket{0}\) and \(\ket{\mathcal{E}_{\text{max}}}\) does not affect the analysis since the added levels remain invariant under this evolution.
\end{proof}
 While it is expected that the correspondence holds for higher dimensional targets as well, it requires an analysis of the master equation that we leave for further work. The zero coupling of the target to the environment needed for exact correspondence comes from the fact that in the repeated cycles paradigm, there is never any reheating between the cycles, which corresponds to an open quantum system in which the target is not coupled to a bath. Any realistic coupling would only worsen the bound. The correspondence of Theorem \ref{th4} links extremal paradigms of control in quantum thermal machines, perfect batteries and well-timed operations on the one hand, and autonomous machines with no external source of work or timing control on the other, demonstrating that increased control on a quantum system does, after all, not significantly impact the cooling performance in this context, but that the difference manifests as different work costs and the actual challenge rather lies in designing appropriate interaction Hamiltonians for autonomous machines \cite{automitch}.

\textit{Smallest machines for maximal cooling.}
Interestingly the smallest possible implementation can already attain the bounds. For a qubit target of energy gap \(E_S\), the simplest coherent machine consists of a single qubit, of energy gap \(E_M\). The Hilbert space of the joint target and machine system is spanned by \(\{\ket{ij}_{SM}\}_{i,j=0,1}\). 
Maximal cooling is achieved by a single unitary operation, swapping the states $\ket{01}_{SM}$ and $\ket{10}_{SM}$. The final state of the target thus has the same populations and Gibbs ratio as the initial state of the machine, i.e.
\begin{align}\label{eq:cohqubittemp}
\beta^* E_S &= \beta_R E_M,
\end{align}
which matches Eq. \eqref{math:besttempqubit} for maximal cooling with $E_M = \mathcal{E}_{max}$. {This single qubit machine is also sufficient to cool higher dimensional targets to the bound \eqref{math:beststate}, using the coherent max-swap protocol (Theorem 2).}

Note that if the machine had additional energy levels in between $\ket{0}_M$ and $\ket{1}_M$, then a single swap would not be sufficient to recover Eq.~\eqref{eq:cohqubittemp} for a qubit target, and one would only do so in the limit of infinite cycles. {Furthermore, higher dimensional targets would also be cooled at a slower rate (per cycle).} In this sense, a single qubit machine is more effective than a more complex machine. The advantage of a more complex machine is revealed when analyzing the work cost of the operation (see \cite{ourPRE}), and in general can help increase $\mathcal{E}_{max}$ by composition.

In the incoherent paradigm, a single qubit machine is unable to cool. The simplest machine allowing for cooling features an extra qubit of energy $E_A = E_M - E_S$, as proposed in Theorem~\ref{thInattain}. This corresponds to the smallest autonomous refrigerator \cite{auto1}. The optimal unitary for cooling is now the swap between the states $\ket{010}_{SMA}$ and $\ket{101}_{SMA}$. Comparing this to the single qubit coherent machine, we see that the swap is identical w.r.t. the target \((S)\) and the original machine \((M)\), and the role of the additional qubit \((A)\) is to enable this transition for incoherent cooling by bridging the energy gap between \(S\) and \(M\). In coherent cooling, this role is implicitly fulfilled by the battery that allows for arbitrary unitary operations.

Cooling over multiple cycles consists in repeatedly thermalising $M$ to $T_R$ and $A$ to $T_H$ and performing the relevant swap. In the limit of infinite cycles, the temperature achieved by the target is given by 
\begin{align}\label{eq:incohqubittemp}
	\beta^*_{\text{inc}} E_S &= \beta_R E_M - \beta_H (E_M - E_S).
\end{align}
In the limit of \(T_H \rightarrow \infty\), the second term on the RHS vanishes and we get Theorem \ref{thInattain}. 

\textit{Conclusion.} We derive universal and attainable bounds for cooling using any quantum machine and target system. 
Interestingly, this bound only depends on the largest energy gap present within the quantum refrigerator and is independent of all other spectral properties and valid for all temperature regimes. For qubits, this bound is already attainable by the simplest possible machines: either by a single cycle of a coherent one-qubit machine or, in the continuous autonomous quantum refrigerator paradigm, as a steady state of a two-qubit machine.
These results unify different operational approaches to quantum thermodynamics and thus go beyond one particular approach. They embody one of the central conceptual pillars of statistical physics, that, despite the potential complexity, thermodynamic tasks can be characterised by a small number of system parameters that need no detailed knowledge of the microstates. The universal bounds and the attainability protocols presented here, are, in general, of course highly idealised and go beyond realistic control over many-body quantum systems. That makes the attainability by few qubit machines all the more interesting, however, as they could potentially be realised with state-of-the-art quantum technologies. A more detailed analysis of few qubit machines in both paradigms, including a finite number of cycles and the respective work costs of achieving the temperature bound can be found in \cite{ourPRE}.
Future investigations should include a trade-off between complexity and achievable $\mathcal{E}_{max}$, as the required unitaries would quickly become impossible to even approximate in the regime of large machines.  Another question beyond asymptotics is the actual convergence rate. We showcase the simplest case in the appendix and hope to gain more in-depth insight in the future. One could also further extend the operational paradigm, by including quantum measurements of the target or working fluid, such as in \cite{qmc}, keeping in mind that there, one also implicitly assumes large measurement machines that replace the refrigerator \cite{YelNico}.  Or one could investigate the advantage of more general thermal couplings, such as in \cite{Nayeli}. Finally, in the cooling task considered here, both initial and final states are diagonal in the energy eigenbasis, such that further limitations from the manipulation of coherences do not apply \cite{coherence1,coherence2}. It could be interesting to study how initial coherences affect the results.

\textit{Acknowledgments} We are grateful to F. Hirsch, P. P. Hofer, M.-O. Renou, and T. Krivachy for fruitful discussions. We would also like to acknowledge all referees for productive and challenging comments. Several authors acknowledge the support from the Swiss NFS and the NCCR QSIT: R.S. through the grants Nos. $200021\_169002$ and $200020\_165843$, N.B. through the Starting grant DIAQ and grant $200021\_169002$, F.C. through the AMBIZIONE grant PZ00P2$\_$161351, and G.H. through the PRIMA grant PR00P2$\_$179748 and Marie-Heim V\"ogtlin grant 164466. MH acknowledges support from the Austrian Science Fund (FWF) through the START project Y879-N27. JBB acknowledges support from the Independent Research Fund Denmark.

\appendix
\onecolumngrid

\section{Cooling bound\label{sec:bound}}

In this section, we give more details on the proof of {\color{black}Theorem \ref{thm:unibound} of the main text (see Claim \ref{theorem:coherentmajorization}) as well as comment on other notions of cooling.}

\begin{claim}\label{theorem:coherentmajorization}
	Under single coherent operations, if the state of the system is initially majorized by the state \(\rho_S^*\) of Theorem 1, its final state is also majorized by \(\rho_S^*\).	
\end{claim}

\begin{proof}

Recall that at the beginning of a single coherent operation the machine is thermal with respect to the  environment temperature $T_R$, and that one can perform any joint unitary on the system and machine. The final state of the system under a single operation is thus
\begin{align}\label{app:basiccoherentoperation}
	\rho_S^\prime &= \Tr_M \left( U \; \left( \rho_S \otimes \tau_M(T_R) \right) \; U^\dagger \right).
\end{align}
	
For the following, we denote the vector of eigenvalues of any state $\rho$ by $\vec{\rho}$, where the eigenvalues are arranged in decreasing order. Consider the problem of optimizing the $k^{th}$ partial sum of $\vec{\rho}_S^{\,\prime}$, over all possible joint unitaries on the system and machine,
\begin{align}
	\max_{U_{SM}} s_k, \quad \text{where} \; s_k &= \sum_{i=0}^{k-1} \left[\vec{\rho}_S^{\,\prime}\right]_i, \quad k \in \{1,2,...,d_S\}.
\end{align}
This is equivalent to optimizing the $k^{th}$ partial sum of the diagonal elements of $\rho^\prime_S$ in a fixed basis, as one can always perform the local unitary that leaves the system diagonal with its eigenvalues arranged in decreasing order along the diagonal (this is the optimal choice to optimise the partial sums, from the Schur-Horn theorem). Indeed, denoting the latter by \(U_S\), we have
\begin{equation}
\text{Tr}_M(U_S \otimes \mathds{1} \rho' U_S^{\dagger} \otimes \mathds{1}) = U_S \text{Tr}_M(\rho') U_S^{\dagger}.
\end{equation}
	
Expressing the final joint state of the system and machine $\rho_{SM}^\prime$ in a product basis (where the choice of basis for the system is the same as above, and for the machine is arbitrary), each diagonal element of $\rho_S^\prime$ is the sum of the corresponding $d_M$ diagonal elements of $\rho_{SM}^\prime$. Thus the optimization of $s_k$ is equivalent to optimizing the partial sum of the first $k\cdot d_M$ diagonal elements of $\rho_{SM}^\prime$, which we shall denote by $\Sigma_k$.
	
However, by the Schur-Horn theorem, the vector of diagonal elements of $\rho^\prime_{SM}$ is majorized by the vector of eigenvalues $\vec{\rho}_{SM}^{\,\prime}$, which in turn is the same as that of the initial state $\vec{\rho}_{SM}$.	As $\rho_{SM}$ is a product state between the system and the thermal state of the machine, its eigenvalues are the products
\begin{align}
	\lambda_{ij} &= \left[\vec{\rho}_S\right]_i \cdot \left[\vec{\tau}\right]_j, \quad i \in \{0,2,...,d_S-1\}, \; j \in \{0,2,...,d_M-1\}.
\end{align}
Arranging these in decreasing order, and summing over the largest $k \cdot d_M$ of them, one has an expression of the form
\begin{align}
	\Sigma_k &= \left[\vec{\tau}\right]_1 \sum_{i=0}^{n_1} \left[\vec{\rho}_S\right]_i + \left[\vec{\tau}\right]_2 \sum_{i=0}^{n_2} \left[\vec{\rho}_S\right]_i + \left[\vec{\tau}\right]_3 \sum_{i=0}^{n_3} \left[\vec{\rho}_S\right]_i + ...,
\end{align}
i.e. the largest eigenvalue of the machine with a subset of the largest eigenvalues of the system, plus the second largest eigenvalue of the machine with a possibly smaller subset of the largest system eigenvalues, and so on ($n_1 \geq n_2 \geq ...$).

Moreover, as $\vec{\rho}_S \prec\vec{\rho}_S^*$, we have from Corollary 1.2 of Ref.~\cite{bondar} that \(\vec{\rho}_S \otimes \vec{\tau}_M \prec\vec{\rho}_S^*\otimes \vec{\tau}_M\) and so \(\Sigma_k\) is upper bounded by the sum of the \(k \cdot d_M\) largest eigenvalues of \(\vec{\rho}_S^*\otimes \vec{\tau}_M\). 
However, the ordering of the eigenvalues of $\vec{\rho}_S^* \otimes \vec{\tau}_M$ is very simple. {\color{black} Recall that the eigenvalues of $\rho_S^*$ are related by
\begin{equation}\label{app:gratio}
  g:=  e^{- \beta_R \mathcal{E}_{max}}=\frac{\left[\vec{\rho}_S^*\right]_{i+1}}{\left[\vec{\rho}_S^*\right]_i}, \quad i= 0 , \dots, d_S-2,
\end{equation}
where $\mathcal{E}_{max}$ is the largest energy gap in the machine spectrum.} Denoting the eigenvalues of the product state $\rho_S^*\otimes \tau_M$ by $\mu_{ij}$, we have
\begin{align}
	\frac{\mu_{ij}}{\mu_{i+1,j^\prime}} &= \frac{ \left[\vec{\rho}_S^*\right]_i \cdot \left[\vec{\tau}\right]_j }{ \left[\vec{\rho}_S^*\right]_{i+1} \cdot \left[\vec{\tau}\right]_{j^\prime} } = \frac{1}{g} \cdot \frac{\left[\vec{\tau}\right]_j}{\left[\vec{\tau}\right]_{j^\prime}} = \frac{1}{g} e^{-(\mathcal{E}_j  - \mathcal{E}_{j^\prime})/T} \geq \frac{1}{g} e^{-\mathcal{E}_{max}/T} = 1, \quad \forall i \in \{0,\dots, d_S-2\}, j,j^\prime \in \{0, \dots, d_M-1\}.
\end{align}
Thus the ordering of eigenvalues of $\vec{\rho}_S^* \otimes \vec{\tau}_M$ in decreasing order begins with $\left[\vec{\rho}_S^*\right]_0$ times all the eigenvalues of the machine, then $\left[\vec{\rho}_S^*\right]_1$ times all of the eigenvalues, and so on. The partial sum of the $k \cdot d_M$ largest eigenvalues is then
\begin{align}
	\left[\vec{\rho}_S^*\right]_0 \sum_{j=0}^{d_M-1} \left[\tau\right]_j + \left[\vec{\rho}_S^*\right]_1 \sum_{j=0}^{d_M-1} \left[\tau\right]_j + ... + \left[\vec{\rho}_S^*\right]_{k-1} \sum_{j=0}^{d_M-1} \left[\tau\right]_j = \sum_{i=0}^{k-1} \left[\vec{\rho}_S^*\right]_i.
\end{align}

Stringing together the inequalities we have gone through, we have proven that
\begin{align}
	s_k &= \sum_{i=0}^{k-1} \left[\vec{\rho}_S^{\,\prime}\right]_i \leq \Sigma_k \leq \sum_{i=0}^{k-1} \left[\vec{\rho}_S^*\right]_i, \quad k \in \{0,2,...,d_S-1\},
\end{align}
that is to say, $\vec{\rho}_S^* \succ \vec{\rho}_S^{\,\prime}$, as desired.
\end{proof}

{\color{black} 

As suggested in the main text, $\rho^*_S$ is actually also the coldest state reachable via coherent operations with respect to other notions of cooling. For any Schur-concave or -convex function such as the von Neumann entropy or the purity this directly follows from the fact that $\rho^*_S$ majorizes any other state attainable by coherent operations. For the mean energy, since that measure depends on the ordering of the diagonal entries, one needs to additionally use the fact that $\rho^*_S$ is passive. With that one can easily show that $\rho_S^*$ has an average energy that lower bounds that of any final state in the coherent paradigm which is what we want to prove next.

\textit{Proof.} Let $\sigma_S$ be a state majorized by $\rho_S^*$ and denote by $\vec{\text{diag}}(\sigma_S)$ the vector of diagonal entries of $\sigma_S$ in the energy eigenbasis. Recall that $\rho_S^*$ is diagonal in the energy basis, and ordered so that it's eigenvalues $\left[\rho_S^*\right]_i$ are decreasing w.r.t. increasing energy $E_i$. Its average energy is then
\begin{align}
   \text{Tr}\left( \rho_S^* H\right) &= \sum_{i=0}^{d_S-1} \left[ \rho_S^* \right]_i E_i\\
        &= E_0 \sum_{i=0}^{d_S-1} \left[ \rho_S^* \right]_i + \left( E_1 - E_0 \right) \sum_{i=1}^{d_S-1} \left[ \rho_S^* \right]_i + \left( E_2 - E_1 \right) \sum_{i=2}^{d_S-1} \left[ \rho_S^* \right]_i + ... + \left( E_{d_S-1} - E_{d_S-2} \right) \left[ \rho_S^* \right]_{d_S-1} \\
        &\leq E_0 \sum_{i=0}^{d_S-1} \left[ \sigma_S \right]_{ii} + \left( E_1 - E_0 \right) \sum_{i=0}^{d_S-1} \left[ \sigma_S \right]_i + \left( E_2 - E_1 \right) \sum_{i=0}^{d_S-1} \left[ \sigma_S \right]_{ii} + ... +\left( E_{d_S-1} - E_{d_S-2} \right) \left[ \sigma_S \right]_{d_S-1} \label{eq:energyinter}\\
        &= \text{Tr}(\sigma_S H),
\end{align}
where $\left[ \sigma_S \right]_{ii}$ are the diagonal elements of $\sigma_S$ and (\ref{eq:energyinter}) holds since from $\vec{\rho}_S^* \succ \vec{\sigma}_S \succ \vec{\text{diag}}(\sigma_S)$, where $\vec{\sigma}_S \succ \vec{\text{diag}}(\sigma_S)$ follows from the Schur-Horn Theorem, it follows that 

\begin{equation}
    \sum_{i=k}^{d_S} \left[ \rho_S^* \right]_i \leq \sum_{i=k}^{d_S} \left[ \sigma_S \right]_{ii} \quad \forall k \in \{ 1,2,...,d_S \}.
\end{equation}

}

\section{Coherent attainability}\label{sec:cohattain}

In this section we provide more details about the proof of Theorem \ref{thCoattain} of the main text. As in the previous Section we denote the vector of eigenvalues of $\rho$ as $\vec{\rho}$, where the eigenvalues are arranged in decreasing order. In the main text we discussed the convergence of two different cooling protocols, both of which we proceed to formally define.

\begin{definition}[Optimal coherent protocol]

Given a joint state $\rho_{SM}$ let $U_{\text{opt}}$ be the unitary that reorders (and potentially rotates) the eigenvalues of $\rho_{SM}$ as largest in the energy subspace $\ket{00} \bra{00}_{SM}$, second largest in $\ket{0\mathcal{E}_1} \bra{0\mathcal{E}_1}_{SM}$ and so on all the way up to $\ket{E_{d_S-1} \mathcal{E}_\text{max}} \bra{E_{d_S-1} \mathcal{E}_\text{max}}_{SM}$. That is 
\begin{equation}\label{eq:optcohunitary}
    U_{\text{opt}} \rho_{SM} U_{\text{opt}}^{\dagger} := \sum_{\substack{i \in \{0, \dots, d_S-1\},\\j \in \{0, \dots, d_M-1\}}} [\vec{\rho}_{SM}]_{i \cdot d_M +j} \ket{E_i \mathcal{E}_j} \bra{E_i \mathcal{E}_j}.
\end{equation}

The optimal coherent protocol is then defined as applying $A$ to the system state in each step, where
\begin{equation}
   \rho_S \mapsto A(\rho_S):= \text{Tr}(U_{\text{opt}} \rho_S \otimes \tau_M (\beta_R) U_{\text{opt}}^{\dagger}).
\end{equation}
\end{definition}

Here we prove that this protocol converges to \(\rho_S^*\), and so does another simpler protocol, the \emph{coherent max-swap} protocol, which only uses the $\mathcal{E}_{\text{max}}$ information of the machine, i.e. the qubit subspace of the machine consisting of the ground and maximally excited states, $\ket{0}_M$ and $\ket{\mathcal{E}_\text{max}}_M$. This protocol is formally defined as

\begin{definition}[Coherent max-swap protocol]

Given a system state $\rho_S$, let $\bar{k}$ be the index $i \in \{1,\dots,d_S-1\}$ for which $\Delta_i := [\vec{\rho}_S]_{i} [\vec{\tau}_M]_0 - [\vec{\rho}_S]_{i-1} [\vec{\tau}_M]_{d_M-1}$ is the greatest if there exists a positive $\Delta_i$, else let $\bar{k}=0$. That is
\begin{equation}
\bar{k}:= \begin{cases}
\argmax\limits_{i=1,\dots,d_S-1} \Delta_i &, \text{if} \max_i \Delta_i > 0\\
0&, \text{else}.
    \end{cases}
\end{equation}
Let $U_0=\mathds{1}_{SM}$, and for $i=1,\dots, d_S-1$ let $U_i$ be defined as follows.
\begin{equation}
\begin{aligned}
    U_i = \mathds{1}_{SM} &- \ket{E_{i-1} \mathcal{E}_{\text{max}}} \bra{E_{i-1} \mathcal{E}_{\text{max}}} - \ket{E_{i} \mathcal{E}_{0}} \bra{E_{i} \mathcal{E}_{0}}\\
    &  + \ket{E_{i-1} \mathcal{E}_{\text{max}}} \bra{E_{i} \mathcal{E}_{0}} + \ket{E_{i} \mathcal{E}_{0}} \bra{E_{i-1} \mathcal{E}_{\text{max}}}.
    \end{aligned}
\end{equation}

For a given system state $\rho_S$, let $\mathcal{P}(\rho_S)$ be its corresponding passive state, i.e. $\mathcal{P}(\rho_S) := \sum_i [\vec{\rho}_S]_i \ket{E_i} \bra{E_i}$.

The coherent max-swap protocol is then defined as applying $B$ to the system state in each step, where
\begin{equation}
    \rho_S \mapsto B(\rho_S) := \mathcal{P} \left( \text{Tr}_M [ U_{\bar{k}} \mathcal{P}(\rho_S) \otimes \tau_M(\beta_R) U_{\bar{k}}^{\dagger}]\right).
\end{equation}

Note that the transformation $\rho \rightarrow \mathcal{P}(\rho)$ is a unitary, so the above map can be expressed via a single joint unitary on the system and machine, which is an allowed coherent operation.

\end{definition}

The proof of convergence for the coherent max-swap protocol consists of the following 3 steps,

\begin{enumerate}
    \item For each $k=0,\dots,d_S-1$, the $k^{\text{th}}$ partial sum of the state obtained after applying the protocol n times is, as a sequence over n, convergent.
    
    \item The convergent point is a fixed point of the protocol.
    
    \item If the initial state of the system $\rho_S$ satisfies $\rho_S \prec \rho_S^*$, then the convergent point of the protocol is $\rho_S^*$.
\end{enumerate}

The optimal coherent protocol proof of convergence can be done in exactly the same way. After proving that the protocol converges as in point 1., we will however here use the convergence of the coherent max-swap protocol and properties of the optimal coherent protocol instead to show that the converging point is $\rho_S^*$. We now turn to the proof of point 1. for both protocols

\begin{lemma}\label{lemma:cohmaxconv}
The coherent max-swap protocol converges.
\end{lemma}

\begin{proof}
The unitary $U_i$ in the coherent max-swap protocol corresponds to swapping the pair of levels $\{i-1,i\}$ of the system with the maximum energy gap of the machine and the way $\bar{k}$ is chosen, if $\bar{k} \neq 0$, ensures that this swap increases the ($\bar{k}-1$) population of the system. It thus follows that for each $k=1,\dots,d_S-1$, $(\sum_k^{(n)})_{n \in \mathds{N}}$, where $\sum_k^{(n)}$ is the sum of the k greatest eigenvalues of $\rho_S^{(n)}:=B^n(\rho_S)$, is monotonically increasing. As these sequences are each bounded by 1, by the monotone convergence theorem, they respectively converge to $\sum_k^{(\infty)}$. The protocol thus converges to the state
\begin{equation}
 \rho_S^{(\infty)} = \sum_{k=0}^{d_S-1} \underbrace{(\Sigma_{k+1}^{(\infty)}-\Sigma_{k}^{(\infty)})}_{\lambda_k} \ket{E_k} \bra{E_k}.
\end{equation}
\end{proof}

\begin{lemma}\label{lemma:cohoptconv}
The optimal coherent protocol converges.
\end{lemma}

\begin{proof}
The proof is completely analogous to that of Lemma \ref{lemma:cohmaxconv} since the unitary $U_{\text{opt}}$ also only increases the sum of the k greatest eigenvalues of $A^{n}(\rho_S)$. 
\end{proof}
Before we go on we would like to state a corollary that will be useful later
\begin{corollary}\label{cor:Bmaj}
Given any system state $\rho_S$,
\begin{align}
    B^n(\rho_S) \succ B^m(\rho_S) \quad \text{and} \quad  A^n(\rho_S) \succ A^m(\rho_S) \quad \forall n \geq m,\; n,m \in \{0, 1, \dots, \infty\},
\end{align}
where $A^0(\rho_S)=B^0(\rho_S) := \rho_S$.
\end{corollary}
\begin{proof}
This is a direct consequence of the fact that $(\sum_k^{(n)})_{n \in \mathbb{N}}$ are monotonically increasing for both protocols.
\end{proof}

In order to distinguish the a priori different converging points of both protocols, we will in the following denote the converging point of the optimal coherent protocol by $A^{\infty}(\rho_S)$ and the converging point of the coherent max-swap protocol by $B^{\infty}(\rho_S)$. We next turn to proving point 2. for the coherent max-swap protocol. To say that $B^{\infty}(\rho_S)$ is a fixed point is the formal way of stating that the converging point cannot be cooled further. This may at first sight seem obvious but there are protocols for which this fails to be true, for instance if the protocol is a discontinuous function at the converging point. Since our protocol includes the unitary  $U_{\bar{k}}$ which is state dependent, a continuity proof is not straightforward. We therefore directly prove that the convergent point is a fixed point.

\begin{lemma}\label{lemma:fixpointcohmax}
For any initial state $\rho_S$, $B^{\infty}(\rho_S)$ is a fixed point of the coherent max-swap protocol.
\end{lemma}

\begin{proof}
The proof is by contradiction. Assume that $B^{\infty}(\rho_S)$ is not a fixpoint. Looking at the definition of the coherent max-swap protocol one sees that applying B to $B^{\infty}(\rho_S)$ amounts to selecting the index $\bar{k}_{\infty} \in \{ 1, \dots, d_S-1\}$ and applying the unitary $U_{\bar{k}_{\infty}}$ to $B^{\infty}(\rho_S)$, moving the population $\Delta_{\bar{k}_{\infty}}$ from the energy level $\bar{k}_{\infty}$ to $\bar{k}_{\infty}-1$. That $B^{\infty}(\rho_S)$ is not a fixed point means that there exists some real number \(\delta > 0\) such that 
\begin{equation}
\Delta_{\bar{k}_{\infty}}^{(\infty)} = \delta.
\end{equation}

Let $\epsilon := \frac{\delta}{d_S+[\vec{\tau}_M]_0+[\vec{\tau}_M]_{d_M-1}}$. As $\lim_{n \rightarrow \infty} B^n(\rho_S) = B^{\infty}(\rho_S)$ there exists an $n \in \mathds{N}$ for which

\begin{equation}\label{eqs:epsilonball}
    \left \lVert \vec{B}^n(\rho_S)- \vec{B}^{\infty}(\rho_S) \right \rVert_{\text{max}} := \max_i \left \lvert [\vec{B}^n(\rho_S)]_i- [\vec{B}^{\infty}(\rho_S)]_i \right \rvert < \epsilon.
\end{equation}

Convergence happens with respect to any metric since we are in finite dimension, and we picked the max metric out of convenience. From the above follows that 
\begin{align}
[\vec{B}^n(\rho_S)]_{\bar{k}_{\infty}} &> [\vec{B}^{\infty}(\rho_S)]_{\bar{k}_{\infty}} - \epsilon,\label{eqs:epsilon-}\\
-[\vec{B}^n(\rho_S)]_{\bar{k}_{\infty}-1} &>-\left( [\vec{B}^{\infty}(\rho_S)]_{\bar{k}_{\infty}-1} + \epsilon \right).\label{eqs:epsilon+}
\end{align}

From Eqs.~\ref{eqs:epsilonball} also follows that for any $k \in \{0, 1, \dots, d_S-1\}$
\begin{equation}
     \left \lvert \sum_{i=0}^k [\vec{B}^n(\rho_S)]_i- \sum_{i=0}^k [\vec{B}^{\infty}(\rho_S)]_i \right \rvert  \leq \sum_{i=0}^k \left \lvert  [\vec{B}^n(\rho_S)]_i- [\vec{B}^{\infty}(\rho_S)]_i \right \rvert < (k+1) \epsilon \leq d_S \epsilon
\end{equation}

from which we have

\begin{equation}\label{eqs:epsilonds}
    \sum_{i=0}^k [\vec{B}^n(\rho_S)]_i +  d_S \epsilon > \sum_{i=0}^k [ \vec{B}^{\infty} (\rho_S)]_i.
\end{equation}

Now we look at the application of B to $B^n(\rho_S)$. As before, this amounts to selecting the index $\bar{k}_n$ and applying a unitary that will move $\Delta_{\bar{k}_n}$ population fom the energy level $\bar{k}_n$ to $\bar{k}_n-1$. Let $\delta'$ be the amount of population moved from $\bar{k}_{\infty}$ to $\bar{k}_{\infty}-1$ when applying $U_{\bar{k}_{\infty}}$ to $B^n(\rho_S)$, i.e.

\begin{equation}
    \delta':= [\vec{B}^n(\rho_S)]_{\bar{k}_{\infty}} [\vec{\tau}_M]_0 - [\vec{B}^n(\rho_S)]_{\bar{k}_{\infty}-1} [\vec{\tau}_M]_{d_M-1}.
\end{equation}

By definition of $\bar{k}_n$ we have $ \Delta_{\bar{k}_n} \geq \delta'$. Furthermore using Eqs.~\ref{eqs:epsilon-} and Eqs.~\ref{eqs:epsilon+} we find

\begin{align}
    \delta'&\geq \left( [\vec{B}^{\infty}(\rho_S)]_{\bar{k}_{\infty}}-\epsilon \right) [\vec{\tau}_M]_0 - \left( [\vec{B}^{\infty}(\rho_S)]_{\bar{k}_{\infty}-1}+\epsilon \right) [\vec{\tau}_M]_{d_M-1}\\
    &= \delta - \epsilon \left([\vec{\tau}_M]_0+ [\vec{\tau}_M]_{d_M-1} \right)\\
    &= \frac{ \delta d_S}{d_S+[\vec{\tau}_M]_0+[\vec{\tau}_M]_{d_M-1}}=d_S \epsilon.
\end{align}
We therefore have
\begin{equation}\label{eqs:Deltaepsilon}
    \Delta_{\bar{k}_n} \geq \delta' \geq d_S \epsilon.
\end{equation}

This brings us to our desired contradiction as

\begin{align}
    \sum_{i=0}^{\bar{k}_n} [ B^{n+1}(\rho_S)]_i &\stackrel{\text{Def.}}{=} \sum_{i=0}^{\bar{k}_n} [B^n (\rho_S)]_i + \Delta_{\bar{k}_n}\\
    &\stackrel{(\ref{eqs:Deltaepsilon})}{\geq} \sum_{i=0}^{\bar{k}_n} [B^n (\rho_S)]_i + \epsilon d_S\\
    &\stackrel{(\ref{eqs:epsilonds})}{>} \sum_{i=0}^{\bar{k}_n} [B^{\infty}(\rho_S)]_i,
\end{align}

which is a contraction to $B^{n+1} (\rho_S) \prec B^{\infty}(\rho_S)$ (Corollary~\ref{cor:Bmaj}).
\end{proof}

Next we would like to prove what the name of the optimal coherent protocol suggests, namely that it cools more than any other protocol. This is the crucial property of the optimal coherent protocol that will help us proving it converge to $\rho_S^*$.

\begin{lemma}\label{lemma:optimalisoptimal}

The optimal coherent protocol cools more than any other choice of coherent operations. That is,
\begin{equation}
    A \circ A \circ ... \circ A (\rho_S) = A^{\otimes n} (\rho_S) \succ C_n \circ C_{n-1} \circ ... \circ C_1 (\rho_S),
\end{equation}
for all $n$ and arbitrary coherent operations $C_i$. (By $\rho_1 \succ \rho_2$ we mean that $\vec{\rho}_1 \succ \vec{\rho}_2$.)

\end{lemma}

\begin{proof}

We prove the statement by induction. We start with the base case $n=1$, i.e. a single coherent operation. From the argument in section A the largest possible value of the partial sum of the $k$ largest eigenvalues of the final state $\rho_S^\prime$ is the sum of the $k*d_M$ largest eigenvalues of the joint state $\rho_{SM}$.

Consider a single application of the optimal coherent protocol, \eqref{eq:optcohunitary}. Tracing out the machine to find the final state of the system, it is diagonal in the energy eigenbasis,
\begin{align}
    \rho^\prime_S &= \sum_{i=0}^{d_S-1} \left( \sum_{j=0}^{d_M-1} [\vec{\rho}_{SM}]_{i \cdot d_M +j} \right) \ket{E_i} \bra{E_i}
\end{align}

Since $\vec{\rho}_{SM}$ is the vector of eigenvalues of the joint state arranged in decreasing order, it follows that the ground state population of the final state of the system, i.e. $\braket{E_0 | \rho_S^\prime | E_0}$, is composed of the sum of the largest $d_M$ eigenvalues of $\rho_{SM}$, and the first excited state is the sum of the next largest eigenvalues of $\rho_{SM}$, and so on. Thus the partial sum of the largest $k$ eigenvalues of $\rho_S^\prime$ is the sum of the $k*d_M$ largest eigenvalues of the joint state. This concludes the proof for $n=1$.

Next we prove that if $\rho_1 \succ \rho_2$, then $A(\rho_1) \succ C(\rho_2)$, for an arbitrary coherent operation $C$. Firstly, if $\rho_1 \succ \rho_2$, then it is also true that $\rho_1 \otimes \tau_M(\beta_R) \succ \rho_2 \otimes \tau_M(\beta_R)$. But then the sum of the largest $k*d_M$ eigenvalues of $\rho_1 \otimes \tau_M(\beta_R)$ is larger than (or equal to) that in the second case. Thus $A(\rho_1) \succ A(\rho_2)$. But we have proven above that $A(\rho_2) \succ C(\rho_2)$. Therefore $A(\rho_1) \succ C(\rho_2)$.

Finally, assume that $A^{\otimes n}(\rho_S) \succ C^{\otimes n}(\rho_S)$. Let $\rho_1=A^{\otimes n}(\rho_S)$ and $\rho_2=C^{\otimes n}(\rho_S)$. Then from the previous argument $A^{\otimes (n+1)}(\rho_S) \succ C^{\otimes (n+1)}(\rho_S)$, which concludes the proof.
\end{proof}

Before proving point 3. for both protocols, we need a last lemma about the fixed points of the max-swap protocol.

\begin{lemma}\label{lemma:fixpointmajmaxcoh}
All of the fixed points of the max-swap protocol have the property of majorizing $\rho^*$.
\end{lemma}

\begin{proof}

From the definition of the max-swap protocol, the fixed points of the protocol are those for which $\Delta_i \leq  0 $ for all $i$. More precisely, denoting the fixed point by $\rho_f$,
\begin{align} 
    \frac{ [\vec{\rho}_f]_{i} }{ [\vec{\rho}_f]_{i-1} } \leq \frac{ [\vec{\tau}_M]_{d_M-1} }{ [\vec{\tau}_M]_0 } = g = \frac{ [\vec{\rho}^*_S]_{i} }{ [\vec{\rho}^*_S]_{i-1} },
\end{align}
using the definition of $\rho^*_S$ from Theorem \ref{thm:unibound} of the main text. By repeated application of the above inequality, we get that
\begin{equation}
    \frac{ [\vec{\rho}_f]_{i} }{ [\vec{\rho}_f]_{i-m} } = \frac{ [\vec{\rho}_f]_{i} }{ [\vec{\rho}_f]_{i-1} } \frac{ [\vec{\rho}_f]_{i-1} }{ [\vec{\rho}_f]_{i-2} } \dots \frac{ [\vec{\rho}_f]_{i-m+1} }{ [\vec{\rho}_f]_{i-m} }\leq \frac{ [\vec{\rho}^*_S]_{i} }{ [\vec{\rho}^*_S]_{i-m} } \forall m > 0.
\end{equation}

Next consider the following ratio of partial sums, for $k \in \{0, \dots, d_S-2\}$,

\begin{equation}
\frac{\sum_{i=k+1}^{d_S-1} [\vec{\rho}_f]_{i}}{\sum_{j=1}^{k} [\vec{\rho}_f]_{j}}= \sum_{i=k+1}^{d_S-1} \frac{[\vec{\rho}_f]_{i}}{\sum_{j=1}^{k} [\vec{\rho}_f]_{j}}= \sum_{i=k+1}^{d_S-1} \frac{1}{\sum_{j=1}^{k} \frac{[\vec{\rho}_f]_{j}}{[\vec{\rho}_f]_{i}}} \leq \sum_{i=k+1}^{d_S-1} \frac{1}{\sum_{j=1}^{k} \frac{[\vec{\rho}^*_S]_{j}}{[\vec{\rho}^*_S]_{i}}}=\frac{\sum_{i=k+1}^{d_S-1} [\vec{\rho}^*_S]_{i}}{\sum_{j=1}^{k} [\vec{\rho}^*_S]_{j}}.
\end{equation}

Finally consider the partial sum of the $k$ largest eigenvalues of $\rho_S$,
\begin{align}
    \sum_{i=0}^{k-1} [\vec{\rho}_f]_{i} &= \frac{ \sum_{i=0}^{k-1} [\vec{\rho}_f]_{i} }{ \sum_{i=0}^{d_S-1} [\vec{\rho}_f]_{i} } = \frac{ \sum_{i=0}^{k-1} [\vec{\rho}_f]_{i} }{ \sum_{i=0}^{k-1} [\vec{\rho}_f]_{i} + \sum_{i=k}^{d_S-1} [\vec{\rho}_f]_{i} }\\
    &= \frac{ 1 }{ 1 + \frac{ \sum_{i=k}^{d_S-1} [\vec{\rho}_f]_{i} }{ \sum_{i=0}^{k-1} [\vec{\rho}_f]_{i} } } \\
    &\geq \frac{ 1 }{ 1 + \frac{ \sum_{i=k}^{d_S-1} [\vec{\rho}^*_S]_{i} }{ \sum_{i=0}^{k-1} [\vec{\rho}^*_S]_{i} } } = \sum_{i=0}^{k-1} [\vec{\rho}^*_S]_{i},
\end{align}
which shows that $\rho_f \succ \rho_S^*$ as desired.
\end{proof}

We are now in a position to prove that the coherent max-swap protocol converges to the desired state.

\begin{lemma} \label{lemma:cohmaxrhostar}
If the initial state of the system \(\rho_S\) satisfies $\rho_S \prec \rho_S^*$ then the convergent point of the coherent max-swap protocol is $\rho_S^*$.
\end{lemma}

\begin{proof}

We have proven in Lemma \ref{lemma:cohmaxconv} that the coherent max-swap protocol converges to a point that we denote by $B^{\infty}(\rho_S)$. From Lemma \ref{lemma:fixpointcohmax} and Lemma \ref{lemma:fixpointmajmaxcoh} we furthermore know that this point is a fixed point of the coherent max-swap protocol and that it majorises $\rho_S^*$. However from Theorem \ref{thm:unibound} of the main text we know that the state of the system under the repeated application of any coherent operations must be majorised by $\rho_S^*$.

Therefore, the only state that the coherent max-swap protocol can converge to is $\rho^*$ itself.

 Indeed since \(B^{\infty}(\rho_S) \prec \rho_S^*\), we have \([B^{\infty}(\rho_S)]_0 \leq [\rho_S^*]_0\). Furthermore, as $B^{\infty} (\rho_S)$ is a fixed point of the coherent max-swap protocol
\begin{align} \label{equ:maxratio}
    \frac{ [\vec{B}^{\infty} (\rho_S)]_{i} }{ [\vec{B}^{\infty} (\rho_S)]_{i-1} } \leq \frac{ [\vec{\tau}_M]_{d_M-1} }{ [\vec{\tau}_M]_0 } = g = \frac{ [\vec{\rho}^*_S]_{i} }{ [\vec{\rho}^*_S]_{i-1} }
\end{align} 
 holds true, from which follows that
\begin{align}\label{eqs:lambdasmallerp}
    [\vec{B}^{\infty} (\rho_S)]_k =[\vec{B}^{\infty} (\rho_S)]_0 \left( \frac{[\vec{B}^{\infty} (\rho_S)]_1}{[\vec{B}^{\infty} (\rho_S)]_0} \frac{[\vec{B}^{\infty} (\rho_S)]_2}{[\vec{B}^{\infty} (\rho_S)]_1} ... \frac{[\vec{B}^{\infty} (\rho_S)]_k}{[\vec{B}^{\infty} (\rho_S)]_{k-1}} \right) \quad \leq \quad  [\vec{\rho}^*_S]_0 \left( e^{-\beta_R \mathcal{E}_{max}} \right)^k = [\vec{\rho}^*_S]_k \quad \forall k.
\end{align}

As both $B^{\infty}(\rho_S)$ and $\rho_S^*$ are normalised states, the sum of their eigenvalues are both equal to $1$, which is, using Eq.~\ref{eqs:lambdasmallerp},  only possible if $[\vec{B}^{\infty} (\rho_S)]_i = [\vec{\rho}^*_S]_i$ for all $i$, i.e. if and only if \(B^{\infty}(\rho_S)=\rho_S^*\).

\end{proof}
Using this last result we can finally prove that the optimal coherent protocol also converges to the desired state.
\begin{lemma}
If the initial state of the system \(\rho_S\) satisfies $\rho_S \prec \rho_S^*$ then the convergent point of the optimal protocol is $\rho_S^*$.
\end{lemma}

\begin{proof}
From Theorem 1 and Lemma \ref{lemma:cohoptconv} we know that 
\begin{equation}
    \sum_{i=0}^{k-1} [A^{\infty}(\rho_S)]_i \leq \sum_{i=0}^{k-1} [\rho_S^*]_i, \quad \forall k= 1,\dots, d_S.
\end{equation}

From lemma \ref{lemma:optimalisoptimal} we also have that
\begin{equation}
    \sum_{i=0}^{k-1} [B^{n}(\rho_S)]_i \leq  \sum_{i=0}^{k-1} [A^{n}(\rho_S)]_i , \quad \forall k= 1,\dots, d_S, \forall n \in \mathds{N}.
\end{equation}
And so

\begin{equation}
     \sum_{i=0}^{k-1} [\rho_S^*]_i= \sum_{i=0}^{k-1} [B^{\infty}(\rho_S)]_i \leq    \sum_{i=0}^{k-1} [A^{\infty}(\rho_S)]_i= \sum_{i=0}^{k-1} [\rho_S^*]_i , \quad \forall k= 1,\dots, d_S.
\end{equation}

Therefore for all $k=1,\dots, d_S$ we have $ \sum_{i=0}^{k-1} [A^{\infty}(\rho_S)]_i=\sum_{i=0}^{k-1} [\rho_S^*]$ and 
\begin{align}
    A^{\infty}(\rho_S)&= \sum_{k=0}^{d_S-1} \left( \sum_{i=0}^{k+1} [A^{\infty}(\rho_S)]_i - \sum_{j=0}^{k} [A^{\infty}(\rho_S)]_j\right) \ket{E_k} \bra{E_k}\\
   & =\sum_{k=0}^{d_S-1} \left( \sum_{i=0}^{k+1} [\rho_S^*]_i - \sum_{j=0}^{k} [\rho_S^*]_j\right) \ket{E_k} \bra{E_k} \\
   &= \rho_S^*.
\end{align}

\end{proof}

\section{Incoherent cooling bound}
Here we provide a proof for the attainability of the cooling bound of Theorem 3 (in the case of incoherent cooling operations). We shall use the proof from the section ``Coherent attainability'' of the Appendix, since many of the steps are the same. The proof will go as follows. We shall
\begin{itemize}
	\item define an incoherent version of the \textit{max-swap protocol}, that involves extending the machine via the addition of a maximum of $d_S-1$ qubits, or alternatively, a single $d_S$ dimensional system.
	\item show, following the same arguments as in section ``Coherent attainability'' of the Appendix, that the repeated application of the incoherent max-swap on the target converges to a state that saturates the bound in Theorem 3, for every single qubit subspace of the target that is comprised of successive energy eigenstates.
\end{itemize}

We begin by defining the extended machine more precisely.

\begin{definition}(extended machine)
Given the machine M with hamiltonian $H_M= \sum_{i=0}^{d_M-1} \mathcal{E}_i \ket{\mathcal{E}_i} \bra{\mathcal{E}_i}$ and the system S with hamiltonian $H_S= \sum_{i=0}^{d_S-1} E_i \ket{E_i} \bra{E_i}$ we define the extended machine $\tilde{M}$ by appending $d_S-1$ qubits of gap
\begin{equation}
    \mathcal{E}_{\text{max}}-(E_i-E_{i-1}), \quad i-1, \dots, d_S-1
\end{equation}
respectively to the machine M. Denoting the ground state population and excited state population of the $i^{\text{th}}$ qubit by $F_i^0$ and $F_i^1$ respectively, the Hamiltonian of the extended machine is therefore given by
\begin{equation}
    H_{\tilde{M}}= H_M \otimes \mathds{1}_{\bar{M}} + \sum_{i=0}^{d_S-1} H_i \otimes \mathds{1}_{\bar{i}},
\end{equation}
with $H_i = F_i^0 \ket{F_i^0} \bra{F_i^0} + F_i^1 \ket{F_i^1} \bra{F_i^1}$.
\end{definition}

Without loss of generality we will set $F_i^0=0$ for all $i=1,\dots,d_S-1$ such that
\begin{equation}
    F_i^1=\mathcal{E}_{\text{max}}-(E_i-E_{i-1}).
\end{equation}

The extension is coupled to the hot bath for every cycle of the incoherent paradigm, and thus each qubit $i$ is in a thermal state at $T_H$, that we denote by $\tau_i^H$, prior to the application of the joint unitary operation between the target and the extended machine.
The state of the extended machine at the beginning of each cycle is therefore given by
\begin{equation}
    \rho_{\tilde{M}}=\rho_M \otimes \tau_1^H \otimes \dots \otimes \tau_{d_S-1}^H.
\end{equation}

The incoherent max-swap protocol is then defined as follows

\begin{definition}[Incoherent max-swap protocol]

Given a system state $\rho_S$ and the extended machine $\tilde{M}$, let $\bar{k}$ be the index $i \in \{1,\dots,d_S-1\}$ for which $\tilde{\Delta}_i := \bra{E_i}\rho_S \ket{E_i} [\vec{\tau}_M]_0 [\vec{\tau}_i^H]_1 - \bra{E_{i-1}}\rho_S \ket{E_{i-1}} [\vec{\tau}_M]_{d_M-1} [\vec{\tau}_i^H]_0$ is the greatest if there exists a positive $\tilde{\Delta}_i$, else let $\bar{k}=0$. That is
\begin{equation}
\bar{k}:= \begin{cases}
\argmax\limits_{i=1,\dots,d_S-1} \tilde{\Delta}_i &, \text{if} \max_i \tilde{\Delta}_i > 0\\
0&, \text{else}.
    \end{cases}
\end{equation}
Let $\tilde{U}_0=\mathds{1}_{S\tilde{M}}$, and for $i=1,\dots, d_S-1$ let $\tilde{U}_i$ be defined as follows.
\begin{equation}
\begin{aligned}
    \tilde{U}_i = \mathds{1}_{S\tilde{M}} &- \ket{E_{i-1} \mathcal{E}_{\text{max}} F_k^0} \bra{E_{i-1} \mathcal{E}_{\text{max}} F_k^1}\otimes \mathds{1}_{\bar{SMk}} - \ket{E_{i} \mathcal{E}_{0} F_k^1} \bra{E_{i} \mathcal{E}_{0}F_k^1}\otimes \mathds{1}_{\bar{SMk}}\\
    &  + \ket{E_{i-1} \mathcal{E}_{\text{max}} F_k^0} \bra{E_{i} \mathcal{E}_{0} F_k^1}\otimes \mathds{1}_{\bar{SMk}} + \ket{E_{i} \mathcal{E}_{0} F_k^1} \bra{E_{i-1} \mathcal{E}_{\text{max}} F_k^0}\otimes \mathds{1}_{\bar{SMk}}.
    \end{aligned}
\end{equation}

The incoherent max-swap protocol is then defined as applying $\tilde{B}$ to the system state in each step, where
\begin{equation}
    \rho_S \mapsto \tilde{B}(\rho_S) :=  \text{Tr}_{\tilde{M}} [ \tilde{U}_{\bar{k}} \rho_S \otimes \rho_{\tilde{M}} \tilde{U}_{\bar{k}}^{\dagger}].
\end{equation}

The above unitary corresponds to swapping the pair of levels $\{\bar{k}-1,\bar{k}\}$ of the target with the maximum energy gap of the machine and the particular qubit subspace in the extension that makes the unitary an energy preserving swap between degenerate states. Indeed
\begin{equation}
    E_{\bar{k}-1} + \mathcal{E}_{\text{max}}= E_{\bar{k}} + F_{\bar{k}}^1.
\end{equation}

\end{definition}

Note that in the incoherent max-swap protocol one cannot apply $\mathcal{P}$ (the unitary that makes the state passive) since the unitary implementing $\mathcal{P}$ may not be energy preserving. As a consequence, we work with the partial sum of the population of the k lowest excited states instead of the partial sum of the k greatest population but does not change the idea of any of the proof as well as any of the results. One can also remove $\mathcal{P}$ in the definition of the coherent max-swap protocol and all the same results would hold. It is just more efficient to reorder the populations in the right order if one can and is the reason we opted for this alternative in the definition of the coherent max-swap protocol. 

\begin{lemma}
The incoherent max-swap protocol converges.
\end{lemma}

\begin{proof}

Let $D^{(n)}_k$ be the sum of the populations of the first k excited states of $\tilde{B}^{n}(\rho_S)$, that is
\begin{equation}
    D^{(n)}_k=\sum_{i=0}^{k-1} \bra{E_i} \tilde{B}^{n}(\rho_S) \ket{E_i}.
\end{equation}

Analogously as for the coherent max-swap protocol, it follows from the definition of the incoherent max-swap protocol that for each $k=1,\dots,d_S$, $(D_k^{(n)})_{n \in \mathds{N}}$ is monotonically increasing. As these sequences are each bounded by 1, by the monotone convergence theorem, they respectively converge to $D_k^{(\infty)}$. The protocol thus converges to the state
\begin{equation}
    \tilde{B}^{\infty}(\rho_S) = \sum_{k=0}^{d_S-1} (D_{k+1}^{(\infty)}-D_{k}^{(\infty)}) \ket{E_k} \bra{E_k}.
\end{equation}
\end{proof}

\begin{lemma}
For any initial state $\rho_S$, $\tilde{B}^{\infty}(\rho_S)$ is a fixed point of the coherent max-swap protocol.
\end{lemma}

\begin{proof}
The proof is completely analogous to that of Lemma \ref{lemma:fixpointcohmax}. One only has to systematically replace $\Delta$ by $\tilde{\Delta}$, and $[\vec{B}^n(\rho_S)]_i$ by $\bra{E_i} \tilde{B}^n(\rho_S) \ket{E_i}$ as well as define $\epsilon$ and $\delta'$ as
\begin{align}
\epsilon &:= \frac{\delta}{d_S+ [\vec{\tau}_M]_0 [\vec{\tau}_{\bar{k}_{\infty}}^H]_1+ [\vec{\tau}_M]_{d_M-1} [\vec{\tau}_{\bar{k}_{\infty}}^H]_0}\\
\delta'&:=\bra{E_{\bar{k}_{\infty}}} \tilde{B}^n(\rho_S)\ket{E_{\bar{k}_{\infty}}} [\vec{\tau}_M]_0 [\vec{\tau}_{\bar{k}_{\infty}}^H]_1+ \bra{E_{\bar{k}_{\infty}-1}} \tilde{B}^n(\rho_S)\ket{E_{\bar{k}_{\infty}-1}}  [\vec{\tau}_M]_{d_M-1} [\vec{\tau}_{\bar{k}_{\infty}}^H]_0.
\end{align}

The contradiction concluding the proof is then to the fact that $(D_k^{(n)})_{n \in \mathbb{N}}$ is monotonically increasing.

\end{proof}

\begin{lemma}
If $\rho_S \prec \rho_S^*$ then $\tilde{B}^{\infty} (\rho_S) \succ \rho_S^*$.
\end{lemma}

\begin{proof}
The fixed points of the incoherent max-swap protocol are those states for which all of the $\tilde{\Delta}_i$ are non-positive. Since the converging point of the protocol $\tilde{B}^{\infty}(\rho_S)$ is a fixed point this means that
\begin{align}
	\frac{\bra{E_i} \tilde{B}^{\infty}(\rho_S) \ket{E_i}}{\bra{E_{i-1}} \tilde{B}^{\infty}(\rho_S) \ket{E_{i-1}}} \leq \frac{ [\vec{\tau}_M]_{d_M-1} }{ [\vec{\tau}_M]_0 } \frac{[\vec{\tau}_{i}^H]_0}{[\vec{\tau}_{i}^H]_1}.
\end{align}

In the case that the temperature of the hot bath is infinite, this is equivalent to \eqref{equ:maxratio}. Thus the protocol must converge to a state with ratios smaller or equal to that of $\rho^*_S$, which from the proof of Lemma \ref{lemma:fixpointmajmaxcoh} means that $\tilde{B}^{\infty}(\rho_S)$ majorizes $\rho^*_S$. 

\end{proof}

\section{convergence rate}

We would like to calculate the convergence rate of the coherent max-swap protocol for an illustrative scenario. We do this for a target qubit system and a machine consisting of $n$ identical qubits, $n$ being a natural number. To focus on the rate of convergence, we keep the final cooling bound independent of $n$, by keeping the largest energy gap of the full machine fixed at $\mathcal{E}_\text{max}$. For a given $n$, each machine qubit is therefore of gap $\frac{\mathcal{E}_{\text{max}}}{n}$. Each step of the algorithm consists in performing a swap between the system qubit and the $\mathcal{E}_{\text{max}}$ virtual qubit of the machine. That corresponds to performing the unitary swapping the energy levels $\ket{0 \mathcal{E}_{\text{max}}}_{SM}$ and $\ket{E_S 0}_{SM}$. Upon performing the swap, the ground state population of the system is changed from $p_0$ to

\begin{equation}
p_0'= p_0 (1-q_{1 \dots 1}) + p_1 q_{0 \dots 0},
\end{equation}

where $p_0$ and $p_1$ are the ground and excited state population of the system and $q_i$, i written in binary, is the $i^{\text{th}}$ excited state population of the machine, when thermalised to $\beta_R$. 

One can rewrite $p_0'$ as

\begin{equation}
p_0'= (p_0 - r_V) (1-N_n)+r_V,
\end{equation}
with $N_n:= q_{0 \dots 0}+q_{1 \dots 1}$ the norm of the $\mathcal{E}_{\text{max}}$ virtual qubit machine and 
\begin{equation}
r_V:= \frac{q_{0 \dots 0}}{N_n}= \frac{1}{1+e^{-\beta_R \mathcal{E}_{\text{max}}}}
\end{equation}
 its normalised ground state population. Note that $r_V$ is independent of n, the number of machine qubits. After performing k steps of the algorithm we therefore have

\begin{equation}
p_0^k-r_V=(p_0-r_V) (1-N_n)^k
\end{equation}
and so $\lim_{k \rightarrow \infty} p_0^k= r_V$. The rate of convergence is therefore given by

\begin{equation}
\lim_{k \rightarrow \infty} \frac{p_0^k - r_V}{p_0^{k-1}-r_V}=1-N_n.
\end{equation}

The last equation holds in fact for any $k \in \mathds{N}$, not only in the limit $k \rightarrow \infty$. In any case, calculating the rate of convergence amounts to calculating $N_n$. The closer that $N_n$ is to $1$, the faster the convergence. Conversely, $N_n$ close to zero implies slower convergence.

Since the machine is a tensor of $n$ qubits, we can explicitly calculate the norm to be
\begin{align}
    N_n = q_{0 \dots 0} + q_{1 \dots 1} &= q_0^n + q_1^n \\
    &= \left( \frac{1}{1 + e^{-\beta_R \mathcal{E}_\text{max} / n}}  \right)^n + \left( \frac{e^{-\beta_R \mathcal{E}_\text{max} / n}}{1 + e^{-\beta_R \mathcal{E}_\text{max} / n}} \right)^n \\
    &= \frac{1 + e^{-\beta_R \mathcal{E}_\text{max}}}{ \left( 1 + e^{-\beta_R \mathcal{E}_\text{max} / n} \right)^n }.
\end{align}

This is a decreasing function w.r.t. $n$ and for large $n$ scales as $O(2^{-n})$, more precisely,
\begin{align}
    \lim_{n \rightarrow \infty} \frac{N_n}{2^{-n}} &= 2 \cosh\left( \frac{ \beta_R \mathcal{E}_\text{max}}{2} \right). 
\end{align}

This is consistent with the statement in the main text that a single qubit machine is more effective that a more complex machine if both have the same maximum energy gap, and thus the same cooling bound. Complexity helps if it can increase $\mathcal{E}_\text{max}$ by composition, and can also decrease the work cost of cooling.

\begin{table}[h!] 
\centering
    \begin{tabular}{|c||ccc|ccc|}
    \hline
  {\backslashbox{\# cycles}{size}}& \multicolumn{3}{|c|}{finite}&\multicolumn{3}{|c|}{$\infty$}\\ \hline    control & coh & & inc&coh& & inc\\\hline
  1 &$\Lambda_{coh}^1$&$\succ$&$\Lambda_{inc}^1$&CPTP&$\succ$&TO\\\hline $\infty$&$\Lambda_{coh}^\infty$&$\cong$&$\Lambda_{inc}^\infty$&CPTP&$=$&CPTP\\
  \hline
    \end{tabular}
    \caption{Here we depict the possible maps induced on the target system in the different regimes. For single cycles, finite coherent machines yield target states majorising finite incoherent machines. For infinite machine size the maps tend to thermal operations (TO)\cite{TO} in the incoherent control paradigm if one keeps the part of the machine connected to the hot bath finite and views it as being part of the system. See also \cite{finiteTO1} and \cite{finiteTO2} for finite dimensional versions of TO. For single cycle coherent control considering arbitrary infinite size machines, one can, making use of Stinespring's dilation, implement any CPTP map on our system \cite{cohattain}, demonstrating that most cooling paradigms only make sense when considering finite range interactions. In the limit of infinite cycles, the two control paradigms unify, at least for qubit target systems, by making use of an ancillary qubit in the incoherent machine.}\label{table}
\end{table}


\begin{thebibliography}{99}
\bibitem{footnote2}  By extending the machine $M$ with $M_1$ we mean that the initial state of the machine is now $\rho^{M}\otimes \rho^{M_1}$, where both $\rho^{M}$ and $\rho^{M_1}$ are thermal. The total Hamiltonian is hence $H = H_S + H_{M} + H_{M_1}$, and again unitaries commuting with $H$ are considered.
\end{thebibliography}
\end{document}